\documentclass[11pt]{article}

\usepackage{fullpage}
\usepackage{amsmath,amsfonts,amsthm}
\usepackage{mathrsfs,xspace}
\usepackage{times}
\usepackage{mathabx}
\usepackage{tikz,graphicx}
\usetikzlibrary{shapes}

\newenvironment{proofof}[1]{\addvspace{\bigskipamount}\noindent{\em Proof of #1. }}%
{\noindent $\Box$\par\addvspace{\bigskipamount}}

\makeatletter
\def\th@plain{%
  \thm@notefont{}
  \itshape 
}
\def\th@definition{%
  \thm@notefont{}
  \normalfont 
}
\makeatother

\newtheorem{prop}{Proposition}

\newtheorem{theorem}[prop]{Theorem}

\newtheorem{corollary}[prop]{Corollary}

\newtheorem{lemma}[prop]{Lemma}

\newtheorem{claim}[prop]{Claim}

\newtheorem{definition}{Definition}[section]

\newtheorem{remark}{Remark}

\newcommand{\deq}{\mathrel{:=}}  

\newcommand{\accz}{\mathsf{ACC}^0}
\newcommand{\nexp}{\mathsf{NEXP}}

\newcommand{\bollobas}{Bollob\'{a}s\xspace}
\newcommand{\pudlak}{Pudl\'{a}k\xspace}

\DeclareMathOperator{\dcc}{D}

\DeclareMathOperator{\smcc}{D^{\parallel}}

\newcommand{\mpj}{\textsc{mpj}\xspace}
\newcommand{\mpjh}{\widehat{\textsc{mpj}}\xspace}

\newcommand{\plr}{\textsc{plr}\xspace}
\newcommand{\erdosrenyi}{Erd\H{o}s-R\'{e}nyi\xspace}
\newcommand{\erdos}{Erd\H{o}s\xspace}
\DeclareMathOperator{\cost}{\mathrm{cost}}
\newcommand{\independent}{\textsc{uncorrelated}\xspace}
\newcommand{\nonboolean}{non-Boolean\xspace}
\newcommand{\iskclique}{\mbox{ is k-clique}}

\newcommand{\cE}{\mathcal{E}}
\newcommand{\cL}{\mathcal{L}}
\newcommand{\cP}{\mathcal{P}}
\newcommand{\cR}{\mathcal{R}}
\newcommand{\cQ}{\mathcal{Q}}

\newcommand{\eps}{\varepsilon}
\newcommand{\xor}{\textsc{xor}}

\newcommand{\ldgraphname}{dependent random graph}
\newcommand{\ldgraph}[3]{G_{#3}(#1,#2)}
\newcommand{\gnpd}{\ldgraph{n}{p}{d}}
\newcommand{\gn}[2]{\ldgraph{n}{#1}{#2}}

\newcommand{\gnp}{\ldgraph{n}{p}{}}

\newcommand{\RR}{\mathbb{R}}
\newcommand{\xv}{\textbf{X}}
\newcommand{\SnotX}[2]{#1\setminus#2}

\DeclareMathOperator{\E}{\mathbb{E}}

\DeclareMathOperator{\polylog}{\mathrm{polylog}}
\DeclareMathOperator{\clique}{\mathrm{clique}}
\DeclareMathOperator{\chromatic}{\mathrm{\chi}}

\DeclareMathOperator{\sm}{\setminus}

\title{Dependent Random Graphs and Multiparty Pointer Jumping}

\author{
  Joshua Brody \qquad Mario Sanchez \\ Swarthmore College \\ {\small joshua.e.brody@gmail.com, msanche1@swarthmore.edu}
}

\date{}

\begin{document}
\maketitle
\begin{abstract}
  We initiate a study of a relaxed version of the standard \erdosrenyi
  random graph model, 
  where each edge may depend on a few other edges.
  We call such graphs \emph{\ldgraphname s}.  Our main result
  in this direction is a thorough understanding of the clique number
  of \ldgraphname s.  We also obtain bounds for the chromatic number.
  Surprisingly, many of the standard properties of random graphs also
  hold in this relaxed setting.  We show that with high
  probability, a \ldgraphname\xspace will contain a clique of size
  $\frac{(1-o(1))\log(n)}{\log(1/p)}$, and the chromatic number will be at
  most $\frac{n \log(1/(1-p))}{\log n}$.  We expect these results to be of
  independent interest.
  As an application and second main result, we give a new communication
  protocol for the $k$-player Multiparty Pointer Jumping ($\mpj_k$)
  problem in the number-on-the-forehead (NOF) model.  Multiparty
  Pointer Jumping is one of the canonical NOF communication problems,
  yet even for three players, its communication complexity is not well
  understood.  Our protocol for $\mpj_3$ costs $O(n (\log \log n)/\log
  n)$ communication, improving on a bound from~\cite{BrodyC08}.  We
  extend our protocol to the non-Boolean pointer jumping problem
  $\mpjh_k$, achieving an upper bound which is $o(n)$ for any $k\geq
  4$ players.  This is the first $o(n)$ protocol for $\mpjh_k$ and
  improves on a bound of Damm, Jukna, and Sgall~\cite{DammJS98}, which
  has stood for almost twenty years.
\end{abstract}
\section{Introduction} \label{sec:intro}
\paragraph{Random Graphs.}
The study of random graphs revolves understanding the following
distribution on graphs: Given $n$ and $p$, define a distribution on
$n$ vertex graphs $G=(V,E)$ by placing each edge $(i,j) \in E$
\textbf{independently} with probability $p$. The first paper on this
topic, authored by \erdos~and~R\'enyi~\cite{ErdosR59}, focused on
connectivity of graphs. Later, \bollobas~and~\erdos~\cite{Bollobas76}
found the interesting result that almost every graph has a clique
number of either $r$ or $r+1$, for some $r \approx \frac{2\log n}{\log
  1/p}$. This remarkable concentration of measure result led to further
investigations of these graphs. Then, \bollobas~\cite{Bollobas88}
solved the question of the chromatic number and showed that almost
every graph has chromatic number $(1+o(1)) \frac{-n\log{(1-p)}}{2\log
  n}$. For more details, consult \bollobas~\cite{Bollobas-rg-book} and
Alon~and~Spencer~\cite{AlonSpencer-book}.

We extend this model by allowing each edge to depend on up to $d$
other edges.  We make no a priori assumptions on \emph{how} the edges
depend on each other except that edges must be independent of all but
at most $d$ other edges.  This defines a family of graph
distributions. We initiate a study of dependent random graphs by
considering the clique number and the chromatic number.  As far as we
know, this is the first work to systematically study such
distributions.
%
However, other relaxations of the standard random graph model have
been studied.  The most relevant for us is that of Alon and
Nussboim~\cite{AlonN08}, who study random graphs where edges are
$k$-wise independent.  \cite{AlonN08} give tight bounds for several
graph properties, including the clique number, the chromatic number,
connectivity, and thresholds for the appearance of subgraphs.
The bounds for k-wise independent graph properties are not as
tight as the standard random graphs, but this is to be expected since
$k$-wise independent random graphs are a family of distributions
rather than a single distribution.  Our dependent random graphs
similarly represent a family of graph distributions.  However,
dependent random graphs are generally not even almost $k$-wise
independent, even for small values of $d$.

\paragraph{NOF Communication Complexity.}  
As an application of our dependent random graphs, we study multiparty
communication problems in the Number-On-The-Forehead (NOF)
communication model defined by Chandra et al.~\cite{ChandraFL83}.  In
this model, there are $k$ players $\plr_1, \cdots, \plr_k$ who wish to
compute some function $f(x_1,\ldots, x_k)$ of their inputs using the
minimal communication possible.  Initially, players share a great deal
of information: each $\plr_i$ sees every input \emph{except}
$x_i$.\footnote{Imagine $x_i$ being written on $\plr_i$'s forehead.
  Then, $\plr_i$ sees inputs on other players' foreheads, but not his
  own.}  Note that a great deal of information is shared before
communication begins; namely, all players except $\plr_i$ see $x_i$.
As a result, for many functions little communication is needed.
Precisely how this shared information affects how much communication
is needed is not currently well understood, even when limiting how
players may communicate.  We consider two well-studied models of
communication.  In the \emph{one-way} communication model, players
each send exactly one message in order (i.e., first $\plr_1$ sends his
message, then $\plr_2$, etc.)  In the \emph{simultaneous-message} (or
SM) model, each player simultaneously sends a single message to a
referee, who processes the messages and outputs an answer.  We use
$\dcc(f)$ and $\smcc(f)$ to denote the communication complexity of $f$
in the one-way and simultaneous-message models respectively.

To date, no explicit function is known which requires a polynomial
amount of communication for $k = O(\polylog n)$ players in the SM
model.  Identifying such a function represents one of the biggest
problems in communication complexity.  Furthermore, a chain of results
\cite{Yao90,HastadG91,BeigelT94} showed that such a lower bound would
place $f$ outside of the complexity class $\accz$.  $\accz$ lies at
the frontier of our current understanding of circuit complexity, and
until the recent work of Williams~\cite{Williams14} it wasn't even
known that $\nexp \nsubseteq \accz$.  The Multipary Pointer Jumping problem is
widely conjectured to require enough communication to place it outside
of $\accz$.  This motivates our study.

\paragraph{The Pointer Jumping Problem.}
There are many variants of the pointer jumping problem.  Here, we
study two: a Boolean version $\mpj_k^n$, and a non-Boolean version
$\mpjh_k^n$.  (From now on, we suppress the $n$ to ease notation).  We
shall formally define these problems in Section~\ref{sec:prelim}, but
for now, each may be described as problems on a directed graph
that has $k+1$ layers of vertices $L_0, \ldots, L_k$.  The first layer
$L_0$ contains a single vertex $s_0$, and layers $L_1,\ldots, L_{k-1}$
contain $n$ vertices each.  In the Boolean version, $L_k$ contains two
vertices, while in the non-Boolean version $L_k$ contains $n$
vertices.  For inputs, each vertex in each layer except $L_k$ has a
single directed edge pointing to some vertex in the next layer.  The
output is the the unique vertex in $L_k$ reachable from $s_0$; i.e.,
the vertex reached by starting at $s_0$ and ``following the pointers''
to the $k$th layer.  Note that the output is a single bit for $\mpj_k$
and a $\log n$-bit string for $\mpjh_k$.  To make this into a
communication game, we place on $\plr_i$'s forehead all edges from
vertices in $L_{i-1}$ to vertices in $L_i$.  If players
speak in any order except $\plr_1, \cdots, \plr_k$, there is an easy
$O(\log n)$-bit protocol for $\mpj_k$.  
%

%
This problem was first studied by Wigderson,\footnote{This was
  unpublished, but an exposition appears in~\cite{BabaiHK01}.} who
gave an $\Omega(\sqrt{n})$ lower bound for $\mpj_3$.  This was later
extended by Viola and Wigderson~\cite{ViolaW07}, who showed that
$\mpj_k$ requires $\tilde{\Omega}(n^{1/(k-1)})$ communication, even
under randomized communication.  On the upper-bounds side, Pudlak et
al.~\cite{PudlakRS97} showed a protocol for $\mpj_3$ that uses only
$O\left(n(\log \log n)/\log n\right)$ communication, but only works
when the input on $\plr_2$'s forehead is a permutation.  Damm et
al.~\cite{DammJS98} show that $D(\mpjh_3) = O(n \log \log n)$ and
$D(\mpjh_k) = O(n \log^{(k-1)} n)$, where $\log^{(r)} n$ is the $r$th
iterated log of $n$.  Building on~\cite{PudlakRS97}, Brody and
Chakrabarti~\cite{BrodyC08} showed $D(\mpj_3) = O\left(n\sqrt{(\log
  \log n)/\log n}\right)$; they give marginal improvements for
$\mpj_k$ for $k > 3$.  Despite the attention devoted to this problem,
the upper and lower bounds remain far apart, even for $k=3$ players,
where $D(\mpj_3) = \Omega(\sqrt{n})$ and $D(\mpj_3) =
O(n\sqrt{(\log\log n)/\log n})$.  For this reason, in this work we
focus on $\mpj_k$ and $\mpjh_k$ for small values of $k$.  We strongly
believe that fully understanding the communication complexity of
$\mpj_3$ will shed light on the general problem as well.

\subsection{Our Results}
We give two collections of results: one for dependent random graphs,
and the other for the communication complexity of $\mpj_k$ and
$\mpjh_k$.  For our work on dependent random graphs, we focus on the
clique number and on the chromatic number.  The clique number of a
graph $G$, denoted $\clique(G)$, is the size of the largest clique;
the chromatic number $\chromatic(G)$ is the number of colors needed to
color the vertices such that the endpoints of each edge have different
colors.  We use e.g. $\clique(\gnpd)$ to refer to $\clique(G)$ for
some $G\sim \gnpd$.  We achieve upper and lower bounds for each graph
property.  Say that a graph property $P$ holds \textbf{almost surely
  (a.s.)} if it holds with probability approaching $1$ as $n$
approaches $\infty$ i.e. if $P$ holds with probability $1-o(1)$.

Our strongest results\footnote{Our choice of $p$ is motivated by what was needed to
  obtain the communication complexity bounds for $\mpj_k$.  We suspect
  that tweaking our technical lemmas will give bounds for any constant
  $p$.} give a lower bound for $\clique(\gnpd)$ and an
upper bound for $\chromatic(\gnpd)$.
\begin{theorem}\label{thm:clique-lb}
  If $0< p < 1/4$ and $d/p << \sqrt{n}$, then $\gnpd$ almost surely
  has a clique of size $\Omega\left(\frac{\log n}{\log 1/p}\right)$.
\end{theorem}
\begin{theorem}\label{thm:chromatic-ub}
  If $3/4<p<1$ and $d = n^{o(1)}$ then almost surely
  $\chromatic(\gnpd) \leq (1+\eps)\frac{-n\log(1-p)}{\log n}$.
\end{theorem}
These bounds nearly match similar results for \erdosrenyi random
graphs.  Our bounds on the other side are not as tight.
\begin{theorem}\label{thm:clique-ub}
  If $0<p<1$ and $d \leq n/\log^2n$, then almost surely
  $\clique(\gnpd) \leq d\log n$.
\end{theorem}
\begin{theorem}\label{thm:chromatic-lb}
  If $0<p<1$ and $d \leq n/\log^2 n$, then almost surely
  $\chromatic(\gnpd) \geq n/(d\log n)$.
\end{theorem}
For large values of $d$, there are wide gaps in the upper and lower
bounds of clique number and chromatic number.  Are these gaps
necessary?  The existing bounds for random graphs show that
Theorems~\ref{thm:clique-lb} and~\ref{thm:chromatic-ub} are close to
optimal.  Our next result witnesses the tightness for
$\clique(\gnpd)$.
\begin{lemma}\label{lem:large-clique}
  For any $d = o(n)$ and any $0<p<1$
  \begin{enumerate}
  \item there are $d$-dependent random graphs that almost surely
    contain cliques of size $\Omega(d)$.
  \item there are $d$-dependent random graphs that almost surely
    contain cliques of size $\Omega(\sqrt{d}\log n)$.
  \end{enumerate}
\end{lemma}
This result shows that Theorem~\ref{thm:clique-ub} is also close to
optimal.  It also demonstrates that tight concentration of measure
does not generally hold for dependent random graphs, even for small
values of $d$.  Nevertheless, we expect that for many specific
dependent random graphs, tight concentration of measure results will
hold.
Finally, we give two simple constructions which
show that with too much dependence, very little can be said about
$\clique(\gnpd)$.
\begin{lemma}\label{lem:larged}
For any $d \geq 2n$, the following statements hold.
  \begin{enumerate}
  \item For any $0<p<1$, there exists a $d$-dependent random
    graph $\gnpd$ that is bipartite with certainty.
  \item For any $1/2 \leq p < 1$, there exists a $d$-dependent random
    graph $\gnpd$ that contains a clique of size $n/2$ with certainty.
  \end{enumerate}
\end{lemma}

\paragraph{Results for Multiparty Pointer Jumping.}
Our main NOF communication complexity result is a new protocol for
$\mpj_3$.
\begin{theorem}\label{thm:mpj-ub}
  $\dcc(\mpj_3) = O(n (\log \log n)/\log n)$.
\end{theorem}
This is the first improvement in the communication complexity of
$\mpj$ since the work of Brody and Chakrabarti~\cite{BrodyC08}.  Next,
we use this protocol to get new bounds for the \nonboolean version.
\begin{theorem}\label{thm:mpjh-ub}
  $\dcc(\mpjh_4) = O\left(n\frac{(\log\log n)^2}{\log n}\right)$.
\end{theorem}
Our protocol for $\mpjh_4$ is the first sublinear-cost protocol for
$\mpjh_k$ for any value of $k$ and improves on the protocol of Damm et
al.~\cite{DammJS98} which has stood for nearly twenty years.
Our last pointer jumping results give upper bounds in the SM setting.
First we show how to convert our protocol from
Theorem~\ref{thm:mpj-ub} to a simultaneous messages protocol.
\begin{lemma}\label{lem:mpj-sm-ub}
  $\smcc(\mpj_3) = O\left(n \frac{\log \log n}{\log n}\right)$.
\end{lemma}
Note that to solve $\mpjh_3$, players can compute each bit of
$f_3(f_2(i))$ using an $\mpj_3$ protocol.  By running $\log n$
instances in parallel, players compute all of
$\mpjh_3(i,f_2,f_3)$.  Thus, we get the following 
bound for $\mpjh_3$.
\begin{corollary}\label{lem:mpjh-sm-ub}
  $\smcc(\mpjh_3) = O(n \log \log n)$.
\end{corollary}
This matches the bound from~\cite{DammJS98} but holds in the more
restrictive SM setting.

\subsection{Obtaining Bounds for Dependent Random Graph Properties}
In this subsection, we describe the technical hook we obtained to
prove our bounds for Theorems~\ref{thm:clique-lb}
and~\ref{thm:chromatic-ub}.  A key piece of intuition is that when
looking at only small subgraphs of $G \sim \gnpd$, the subgraph
usually \emph{looks} like $\gnp$.  This intution is formalized in the
following definition and lemma.
\begin{definition}\label{def:independent}
  Given a dependent random graph $\ldgraph{n}{p}{d}$, call a subset of vertices
  $S \subseteq V$ \emph{\independent} if any two edges in the subgraph
  induced by $S$ are independent.
\end{definition}

\begin{lemma}\label{lem:independent}
  Suppose $d$ and $k$ are integers such that $dk^3 \leq n$.
  Fix any $d$-dependent graph $\ldgraph{n}{p}{d}$, and let
  $S$ be a set of $k$ vertices uniformly chosen from $V$.  Then, we have
  $$\Pr[S \mbox{ is \independent}] \geq 1-\frac{3dk^3}{2n}\ .$$
\end{lemma}

At first glance, it might appear like we are now able to appeal to the
existing arguments for obtaining bounds for $\clique(\gnp)$ and then
$\chromatic(\gnp)$.  Unfortuantely, this is not the case---while most
potential $k$-cliques are \independent, allowing correlation between
edges drives up the \emph{variance}.  In effect, we might expect to
have roughly the same number of $k$-cliques, but these cliques bunch
together.  Nevertheless, we are able to show that when $d$ is small
enough, these cliques don't bunch up too much.  Appropriately bounding
the variance is the most technically involved hurdle in this work,
and is necessary to obtain both the upper bound on the chromatic
number, and the effecient pointer jumping protocol.  
\subsection{Roadmap}
The rest of the paper is organized as follows.  In
Section~\ref{sec:prelim} we specify some notation, give formal
definitions for the problems and models we consider, and provide some
technical lemmas on probability which we'll need in later sections.
We develop our results for dependent random graphs in
Section~\ref{sec:graphs}, deferring some technical lemmas to
Section~\ref{sec:maintechnical}.  We present main result on
Multiparty Pointer Jumping in Section~\ref{sec:mpj}, deferring the
secondary $\mpj_k$ results to Section~\ref{sec:nonbool}. In
Section~\ref{sec:otherstuff} we prove Lemmas~\ref{lem:large-clique}
and~\ref{lem:larged}.

\section{Preliminaries and Notation} \label{sec:prelim} 

We use $[n]$ to denote the set $\{1,\ldots, n\}$, $N$ to denote ${n
  \choose 2}$, and $\exp(z)$ to denote $e^z$.  For a string $x \in
\{0,1\}^n$, let $x[j]$ denote the $j$th bit of $x$.  For a sequence
of random variables $X_0, X_1, \ldots$, we use $\xv_i$ to denote the
subsequence $X_0,\ldots, X_i$.  For a graph $G = (V,E)$,
$\bar{G}$ denotes the complement of $G$.  Given sets $A \subset B
\subset V$, we use $\SnotX{B}{A}$ to denote the set of edges $\{(u,v):
u,v\in B \mbox{ and } \{u,v\} \nsubseteq A\}$.

For a communication problem, we refer to players
as $\plr_1, \ldots, \plr_k$.  When $k=3$, we anthropomorphize players
as Alice, Bob, and Carol.  Our communication complexity measures were
defined in Section~\ref{sec:intro}; for an in-depth development of
communication complexity, consult the excellent standard textbook of
Kushilevitz and Nisan~\cite{KushilevitzNisan-book}.

\subsection{Probability Theory and Random Graphs}
Next, we formalize our notion of dependent random graphs
and describe the tools we use to bound $\clique(\gnpd)$.
\begin{definition}[\cite{DubhashiPanconesi-book}, Definition 5.3]
  A sequence of random variables $Y_0, Y_1, \ldots, Y_n$ is a
  \emph{martingale} with respect to another sequence $X_0, X_1,
  \ldots, X_n$ if for all $i \geq 0$ we have $$Y_i = g_i(\xv_i)$$ for
  some functions $\{g_i\}$ and, for all $i \geq 1$ we have $$E[Y_i |
  \xv_{i-1}] = Y_{i-1}\ .$$
\end{definition}

\begin{theorem}[Azuma's Inequality]
  Let $Y_0, \ldots, Y_n$ be a martingale with respect to $X_0,\ldots,
  X_n$ such that $a_i \leq Y_i - Y_{i-1} \leq b_i$ for all $i\geq 1$.
  Then $$\Pr[Y_n > Y_0 + t], \Pr[Y_n < Y_0 - t] \leq
  \exp\left(-\frac{2t^2}{\sum_i (b_i-a_i)^2}\right)\ .$$
\end{theorem}

Of particular relevance for our work is the \emph{edge-exposure
  martingale}.  Let $G$ be a random graph.  Arbitrarily order possible
edges of the graph $e_1, \ldots, e_N$, and let $X_i$ be the indicator
variable for the event that $e_i \in G$.  Let $f:{n \choose 2}
\rightarrow \RR$ be any function on the edge set, and let $Y_i \deq
E[f(X_1,\ldots, X_N) | \xv_i]$.  It is easy to verify that for any
$f$, $E[Y_i | \xv_{<i}] = Y_{i-1}$, and therefore $\{Y_i\}$ are a
matingale with respect to $\{X_i\}$.  We say that $\{Y_i\}$ is the
edge-exposure martingale for $G$.

It is worth noting that martingales make no assumptions about the
independence of $\{X_i\}$.  We'll use martingales on graph
distributions where each edges may depend on a small number of other
edges.  This notion of \emph{local dependency} is formalized below.

A \emph{dependency graph} for a set of random variables $X =
\{X_1,\ldots, X_N\}$ is a graph $H$ on $[N]$ such that for all $i$,
$X_i$ is independent of $\{X_j: (i,j) \not\in H\}$.  We say that a set
of variables $X$ is {\bf $d$-locally dependent} if there exists a dependency
graph for $X$ where each vertex has degree at most $d$.

The following lemma of Janson~\cite{Janson04} (rephrased in our
notation) bounds the probability that the sum of a series of random
bits deviates far from its expected value, when the random bits have
limited dependence.
\begin{lemma}{\cite{Janson04}}\label{lem:dependent-sum}
  Let $X = \{X_i\}_{i \in [N]}$ be a $d$-locally dependent set of
  identically distributed binary variables, and let $Y = \sum_{i \in
    [N]} X_i$.  Then, for any $t$ we have
 $$\Pr[|Y-\E[Y]| \geq t] \leq e^{\frac{-2t^2}{(d+1)N}}\ .$$
\end{lemma}

For more details and results on probability and concentration of
measure, consult the textbook of Dubhashi and
Panconesi~\cite{DubhashiPanconesi-book}.

\begin{definition}
  A distribution $\ldgraph{n}{p}{d}$ is \emph{$d$-dependent} if each
  edge is placed in the graph with probability $p$, and furthermore
  that the set of edges are $d$-locally dependent.
\end{definition}
Note that taking $d=0$ gives the standard \erdosrenyi graph model.  As
with $k$-wise independent random graphs, $d$-dependent random graphs
are actually a family of graph distributions.  We make no assumptions
on the underlying distribution beyond the fact that each edge depends
on at most $d$ other edges.  We use $\ldgraph{n}{p}{d}$ to denote an
arbitrary \ldgraphname.

A clique in a graph $G = (V,E)$ is a set of vertices $S$ such that the
subgraph induced on $S$ is complete.  Similarly, an independent set
$T$ is a set of vertices whose induced subgraph is empty.  A
\emph{clique cover} of $G$ is a partition of $V$ into
cliques. We let $\clique(G)$ denote the size of the largest clique in
$G$. Let $\chromatic(G)$ denote the chromatic number of $G$; i.e., the
minimum number of colors needed to color the vertex set such that no
two adjacent vertices are colored the same.  Note that $\chromatic(G)$
is the size of the smallest clique cover of $\bar{G}$.

\subsection{Multiparty Pointer Jumping}
Finally, we formally define the Boolean Multiparty Pointer
Jumping function.  Let $i \in [n]$, and let $f_2,\ldots, f_k : [n]^n$,
be functions from $[n]$ to $[n]$.  Let $x \in \{0,1\}^n$.  We define
the $k$-player pointer jumping function $\mpj_k^n : [n] \times \left(
[n]^n\right)^{k-2} \times \{0,1\}^n$ recursively as follows:
\begin{align*}
  \mpj_3^n(i,f_2,x) &\deq x[f_2(i)] \ ,\\ \mpj_k^n(i,f_2,\ldots,
  f_{k-1},x) &\deq \mpj_{k-1}^n(f_2(i), f_3, \ldots, f_{k-1},x)\ .
\end{align*}
The non-Boolean version $\mpjh_k^n : [n] \times
\left([n]^n\right)^{k-1}$ is defined similarly recursively:
\begin{align*}
  \mpjh_3^n(i, f_2, f_3) &\deq f_3(f_2(i)) \ ,\\ \mpjh_k^n(i, f_2,
  \ldots, f_k) &\deq \mpjh_{k-1}^n(f_2(i), f_3, \ldots, f_k)\ .
\end{align*}
Henceforth, we drop the superscript $n$ to ease notation.  Each
problem is turned into a communication game in the natural way.
$\plr_1$ is given $i$; for each $2\leq j < k$, $\plr_j$ receives
$f_j$, and $\plr_k$ receives $x$.  Players must communicate to output
$\mpj_k(i,f_2,\ldots, f_{k-1},x)$.

\section{Dependent Random Graphs} \label{sec:graphs} 
In this section, we prove our main results regarding dependent random
graphs, namely that with high probability they contain a large clique,
and with high probability the chromatic number is not too large.
The two theorems are formally stated below.
\begin{theorem}{(Formal Restatement of Theorem~\ref{thm:clique-lb})}\label{thm:clique}
  For all $0<\eps < 1/4$ there exists $n_0$ such
  that $$\Pr[\clique(\gnpd) > k] > 1-\exp(-n^{1+\eps})\ ,$$ for all
  $n\geq n_0$, for all $n^{-\eps/4} < p < \frac{1}{4}$ and for all
  $d,k$ such that $k \leq \frac{\log(n/(2d\log^3 n))}{\log(1/p)}$ and
  $d/p \leq n^{1/2-\eps}$.
\end{theorem}
This theorem shows $\clique(\gnpd) = \Omega\left(\frac{\log n}{\log
  1/p}\right)$ with high probability, as long as $d/p$ is bounded away
from $\sqrt{n}$.  Furthermore, when $d = n^{o(1)}$, $\clique(\gnpd)
\geq (1-\eps)\frac{\log n}{\log(1/p)}$ with high probability.
\begin{proof}
  This proof follows the classic technique of
  \bollobas~\cite{Bollobas88}, modified to handle dependent random
  graphs.  We need to show that $\gnpd$ contains clique of size $k$.
  To that end, let $Y$ be the largest number of \emph{edge-disjoint}
  \independent $k$-cliques.  First, we give a lower bound on
  $E[Y]$; we defer its proof to Section~\ref{sec:maintechnical}.
  \begin{lemma}\label{lem:Y}
    $\E[Y] \geq \frac{n^2p}{19k^5}.$
  \end{lemma}
  Now, we use the edge-exposure martingale on $\gnpd$ to show that
  with high probability, $Y$ does not stray far from it's expectation.
  %
%
  Let $Y_0,Y_1,\cdots Y_N$, be the edge exposure martingale on
  $\gn{p}{d}$.  Recall that $Y_0 = E[Y], Y_N = Y$, and $Y_i =
  E[Y|\xv_i]$. In a standard random graph model where all edges are
  independently placed in $G$, it is easy to see that conditioning on
  whether or not an edge is in the graph changes the expected number
  of \emph{edge-disjoint} \independent $k$-cliques by at most one.
  This no longer holds when edges are dependent.  However, if the
  graph distribution is $d$-dependent, then conditioning on $X_i$
  changes the expected number of \emph{edge-disjoint} \independent
  $k$-cliques by at most $d$.  Therefore, $\vert Y_{i+1} - Y_i \vert
  \leq d$.
  Then, by Azuma's inequality, Lemma~\ref{lem:Y}, and our assumption
  that $d/p\leq n^{1/2-\eps}$, we have
  \begin{align*}
    \Pr[Y = 0] &\leq \Pr[Y - \mathbb{E}[Y] \leq - \mathbb{E}[Y]]
    \\ &\leq \exp \left ( \frac{-\mathbb{E}[Y]^2}{2Nd^2} \right )
    \\ &= \exp \left(-\frac{n^2p^2}{19^2d^2k^{10}}(1 + o(1)) \right )\\ &
    \leq \exp(-n^{1+ \eps})\ .
  \end{align*}
  Thus, it follows that $\gnpd$ contains an \independent $k$-clique
  with probability at least $1-\exp(-n^{1+\eps})$.  Since every
  \independent clique is still a clique, it is clear
  that $$\Pr[\clique(\gnpd) \geq k] \geq
  1-\exp(-n^{1+\eps})\ .$$
\end{proof}
Next, we use the lower bound on $\clique(\gnpd)$ to obtain an upper
bound on $\chi(\gnpd)$.
\begin{theorem}\label{thm:chromatic}
  For all $0 < \eps < 1/8$ there exists $n_0$ such
  that $$\Pr\left[\chi(\ldgraph{n}{q}{d}) < (1 +
    4\eps)\frac{-n\log(1-q)}{\log n}\right] > 1-\exp(n^{1+\eps})\ ,$$
  for all $3/4 < q < 1-n^{-\eps/4}$, all $d \leq n^{o(1)}$,
  and all $n \geq n_0$.
\end{theorem}
\begin{proof}
  This follows a greedy coloring approach similar
  to~\cite{Bollobas88,PudlakRS97}, but adapted to dependent random
  graphs.  Set $m = \frac{n}{\log^2n}$, $\eps' = 2\eps$, and $p = 1
  -q$.  Let $\cE$ be the event that every induced subgraph $H$ of
  $\gn{q}{d}$ with $m$ vertices has an independent set of size at
  least $k \deq (1-\eps')\frac{\log m}{-\log(1-q)}$.  Independent sets
  in $\gn{q}{d}$ correspond to cliques in the complement graph
  $\overline{\gn{q}{d}}$, which is distributed identically to $\gnpd$.
  Thus, we're able to leverage Theorem~\ref{thm:clique} to bound
  $\Pr[\cE]$.  In particular, since $d \leq n^{o(1)} \leq
  m^{o(1)}$,\footnote{note that $n^\delta = m^{\delta'}$, where
    $\delta' = \delta\frac{\log n}{\log n -2\log\log n}$. If $\delta =
    o(1)$ then $\delta' = o(1)$ as well.} by Theorem~\ref{thm:clique}
  and a union bound we have
  $$\Pr[\cE] > 1 - \binom{n}{m}\exp(-n^{1+\eps'}) >
  1-\exp\left(\frac{n}{\log n} - n^{1+\eps'}\right) > 1 -
  \exp(-n^{1+\eps})\ .$$
  Now, assume $\cE$ holds.  We iteratively construct a coloring for
  $\gn{q}{d}$.  Start with each vertex uncolored.  Repeat the
  following process as long as more than $m$ uncolored vertices
  remain: Select $m$ uncolored vertices.  From their induced subgraph,
  identify an independent set $I$ of size at least $k$.  Then, color
  each vertex in $I$ using a new color.  When at most $m$ uncolored
  vertices remain, color each remaining vertex using a different
  color.  Since two vertices share the same color only if they are in
  an independent set, it's clear this is a valid coloring.  More over,
  for each color in the first phase, we color at least $k >
  (1-\eps')\frac{\log m}{-\log p} > (1-(3/2)\eps)\frac{\log n}{-\log p}$
  vertices.  Hence, the overall number of colors used is at most
  $$\frac{n-m}{(1-(3/2)\eps')(\log n)/(-\log(1-q))} + m \leq
  (1+ 4\eps)\frac{-n\log(1-q)}{\log n}\ .$$
  Therefore, $\chi(\gn{q}{d}) \leq (1 +
  4\eps)\frac{-n\log(1-q)}{\log n}$ as long as $\cE$ holds.  This
  completes the proof.
\end{proof}

Finally, we give an upper bound on $\clique(\gnpd)$ and a lower bound
on $\chi(\gnpd)$, which follow directly from
Lemma~\ref{lem:dependent-sum}.
\begin{theorem}
  For all $0<p<1$ and $d \leq n/\log^2n$, almost surely
  $\clique(\gnpd) = O(d\log n)$.
\end{theorem}
\begin{proof}
  Let $G \sim \gnpd$, and fix some constant $c$ to be determined
  later.  For a set of vertices $S\subseteq V$ of size $|S| = cd\log
  n$, let $BAD_S$ denote the event that $S$ is a clique, and let $BAD
  \deq \bigvee_S BAD_S$.  Note that there are ${n \choose cd\log n} \leq
  exp(cd\log^2 n)$ such events.  Since $G$ is $d$-dependent and
  $S\subset V$, then the subgraph induced by $S$ is also
  $d$-dependent.  Now, define $z \deq {cd\log n \choose 2}$ and let
  $X_1,\ldots, X_z$ be indicator variables for the edges in the
  subgraph induced by $S$.  Finally, let $Y \deq \sum_i X_i$.  Then,
  $E[Y] = pz$, and $BAD_S$ amounts to having $Y = z$.  By
  Lemma~\ref{lem:dependent-sum}, 
  $$\Pr[BAD_S] = \Pr[Y = z] = \Pr[Y-E[Y] \geq z(1-p)] \leq
  \exp\left(-\frac{2z^2(1-p)^2}{(d+1)z}\right) \\ =
  \exp\left(-\frac{2z(1-p)^2}{d+1}\right)\ . $$ Choosing $c =
  1/(1-p)^2$ and using a union bound yields
  \begin{align*}
    \Pr[BAD] &\leq {n\choose z}\Pr[BAD_S] \leq \exp\left(cd\log^2 n -
    \frac{2(1-p)^2}{d+1}(cd\log n)^2\right) \\
    &= \exp\left(cd\log^2 n(1-2c(1-p)^2)\right) \\
    &< \exp(-\Omega(d\log^2 n))\ ,
  \end{align*} 
  Thus, almost surely $\gnpd$ has no clique of size $\geq cd\log n$.
\end{proof}
Our lower bound on $\chromatic(\gnpd)$ follows as a direct corollary,
since any independent set in $\gnpd$ is a clique in the complement
graph $\widebar{\gnpd}$, which is also $d$-dependent.
\begin{corollary}
  If $0<p<1$ and $d \leq n/\log^2 n$, then almost surely
  $\chromatic(\gnpd) \geq n/(d\log n)$.
\end{corollary}

\section{A New Protocol for $MPJ_3$} \label{sec:mpj}
Below, we describe a family of $\mpj_3$ protocols $\{\cP_H\}$
parameterized by a bipartite graph $H = (A\cup B,E)$ with $|A| = |B| =
n$.  In each protocol $\cP_H$, Alice and Bob each independently send a
single message to Carol, who must take the messages and the input she
sees and output $\mpj_3(i,f,x)$.  Bob's communication in each protocol
is simple: given $i$, he sends $x_j$ for each $j$ such that $(i,j) \in
H$.  Alice's message is more involved.  Given $H$ and $f$, she
partitions $[n]$ into \emph{clusters}.  For each cluster in the
partition, she sends the \xor\ of the bits for $x$.  (e.g. if one cluster
is $\{1,3,5\}$, then Alice would send $x[1] \oplus x[3] \oplus x[5]$)
This partition of $[n]$ into clusters is carefully chosen and depends
on $H$ and $f$.  Crucially, it is possible to make this partition so
that for any inputs $i,f$, Bob sends $x[j]$ for each $j$ in the
cluster containing $f(i)$, except for possibly $x[f(i)]$ itself.  We
formalize this clustering below.  Thus, Carol can compute $x[f(i)]$ by
taking the relevant cluster from Alice's message and ``\xor-ing out''
the irrelevant bits using portions of Bob's message. 

Each protocol $\cP_H$ will correctly compute $\mpj_3(i,f,x)$; we then
use the probabilitic method to show that there exists a graph $H$ such
that $\cP_H$ is \emph{efficient}.  At the heart of this probabilistic
analysis is a bound on the chromatic number of a dependent random
graph.  For functions with large preimages, this
dependency becomes too great to handle.

\begin{definition}
  A function $f: [n] \rightarrow [n]$ is \emph{$d$-limited} if 
  $|f^{-1}(j)| \leq d$ for all $j \in [n]$.
\end{definition}

We end up with a protocol $\cP_H$ that is efficient for all inputs
$(i,f,x)$ as long as $f$ is $d$-limited ($d \approx \log n$ suffices);
later, we generalize $\cP_H$ to work for all inputs.

\begin{remark}
  This construction is inspired by the construction of \pudlak et
  al.~\cite{PudlakRS97}, who gave a protocol for $\mpj_3$ that works
  in the special case that the middle layer is a \emph{permutation}
  $\pi$ instead of a general function $f$.  They also use the
  probabilistic method to show that one $\cP_H$ must be efficient.
  The probablistic method argument in our case depends on the
  chromatic number of a dependent random graph; the analysis of the
  permutation-based protocol in~\cite{PudlakRS97} relied on the
  chromatic number of the standard random graph $\gnp$.
\end{remark}

\paragraph{Description of $\cP_H$.}
Let $H = (A\cup B, E)$ be a bipartite graph with $|A| = |B| = n$.
Given $H$ and $f$, define a graph $G_{f,H}$ by placing $(i,j) \in
G_{f,H}$ if and only if both $(i, f(j)) $ and $(j, f(i))$ are in
H.  Let $C_1, \ldots, C_k$ be a clique cover of $G_{f,H}$,
and for each $1 \leq \ell \leq k$, let $S_\ell \deq \{f(j) : j \in
C_\ell\}$.

The protocol $\cP_H$ proceeds as follows.  Given $f$ and $x$, Alice
constructs $G_{f,H}$.  For each clique $C_\ell$, Alice sends $b_\ell
\deq \bigoplus_{j \in S_\ell} x[j]$.  Bob, given $i$ and $x$, sends
$x[j]$ for all $(i,j) \in H$.  We claim these messages enable Carol to
recover $\mpj_3(i,f,x)$.  Indeed, given $i$ and $f$, Carol computes
$G_{f,H}$.  Let $C$ be the clique in the clique cover of $G_{f,H}$
containing $i$, and let $S \deq \{f(j) : j \in C\}$ and $b \deq
\bigoplus_{j \in S} x[j]$.  Note that Alice sends $b$.  Also note that
for any $j \neq i \in C$, there is an edge $(i,j) \in G_{f,H}$.  By
construction, this means that $(i,f(j)) \in H$, so Bob sends
$x[f(j)]$.  Thus, Carol computes $x[f(i)]$ by taking $b$ (which Alice
sends) and ``XOR-ing out'' $x[f(j)]$ for any $j \neq i \in C$.
In this way, $\cP_H$ computes $\mpj_3$.

%
While $\cP_H$ computes $\mpj_3$, it might not do so in a
communication-efficient manner.  The following lemma shows that there
is an efficient protocol whenever $f$ has small preimages.

\begin{lemma}\label{lem:pH-cost}
  For any $d \leq n^{o(1)}$, there exists a bipartite graph $H$ such
  that for all $i\in [n], x \in \{0,1\}^n$, and all $d$-limited
  functions $f$, we have $$\cost(\cP_H) = O\left(n \frac{\log \log
    n}{\log n}\right)\ .$$
\end{lemma}

Before proving Lemma~\ref{lem:pH-cost}, let us see how this
gives the general upper bound.

\begin{theorem}[Restatement of Theorem~\ref{thm:mpj-ub}]
  $\dcc(\mpj_3) = O(n (\log \log n)/\log n)$.
\end{theorem}
\begin{proof}
  Fix $d = \log n$ and let $\cP_H$ be the protocol guaranteed by
  Lemma~\ref{lem:pH-cost}.  We construct a general protocol $\cP$ for
  $\mpj_3$ as follows.  Given $f$, Alice and Carol select a
  $d$-limited function $g$ such that $g(j) = f(j)$ for all $j$ such
  that $|f^{-1}(f(j))| \leq d$.  Note that Alice and Carol can do this
  without communication, by selecting (say) the lexicographically
  least such $g$.  On input $(i,f,x)$, Alice sends the message she
  would have sent in $\cP_H$ on input $(i,g,x)$, along with $x[j]$
  for all $j$ with large preimages.  Bob merely sends the message he
  would have sent in $\cP_H$.  If the preimage of $f(i)$ is large,
  then Carol recovers $x[f(i)]$ directly from the second part of
  Alice's message.  Otherwise, Carol computes $\mpj_3(i,g,x)$ using
  $\cP_H$.  Since $f(i)$ has a small preimage, we know that $x[g(i)] =
  x[f(i)] = \mpj_3(i,f,x)$, so in either case Carol recovers
  $\mpj_3(i,f,x)$.

  The communication cost of $\cP$ is the cost of $\cP_H$, plus one bit
  for each $j$ with preimage $|f^{-1}(j)| > d$.  There are at most
  $n/d$ such $j$.  With $d = \log n$ and using
  Lemma~\ref{lem:pH-cost}, the cost of $\cP$ is
  $$\cost(\cP) \leq \cost(\cP_H) + n/d = O(n (\log \log n)/\log n) +
  O(n/\log n) = O(n (\log \log n)/\log n)\ .$$
\end{proof}
\begin{proofof}{Lemma~\ref{lem:pH-cost}}
  We use the Probabilitstic Method.  Place each edge in $H$
  independently with probability $p = \Theta\left(\frac{\log \log
    n}{\log n}\right)$. Now, for any $d$-limited function $f$, consider
  the graph $G_{f,H}$.  Each edge $(i,j)$ is in $G_{f,H}$ with
  probability $p^2$, but the edges are not independent.  However, we
  claim that if $f$ is $d$-limited, then $G_{f,H}$ is
  ($2d-2$)-dependent.  To see this, note that $(i,j)$ is in $G_{f,H}$
  if both $(i,f(j))$ and $(j, f(i))$ are in $H$.  Therefore, $(i,j)$
  is dependent on (i) any edge $(i,j')$ such that $f(j') = f(j)$, and
  (ii) any edge $(i',j)$ such that $f(i) = f(i')$.  Since $f$ is
  $d$-limited, there are at most $d-1$ choices each for $i'$ and $j'$.
  Thus, each edge depends on at most $2d-2$ other edges, and $G_{f,H}$
  is $(2d-2)$-dependent.

  In $\cP_H$, Alice sends one bit per clique in the clique cover of
  $G_{f,H}$.  Bob sends one bit for each neighbor of $i$ in $H$.
  Thus, we'd like a graph $H$ such that every $i\in [n]$ has a few
  neighbors and every $d$-limited function $f$ has a small clique cover.

  Let $BAD_i$ denote the event that $i$ has more than $2pn$ neighbors
  in $H$.  By a standard Chernoff bound argument, $\Pr[BAD_i] \leq
  \exp(-np^2/2)$.  Next, let $BAD_f$ be the event that at least
  $(1+\eps)\frac{-n\log(p^2)}{\log n}$ cliques are needed to cover the
  vertices in $G_{f,H}$.  Note that any clique in $G_{f,H}$ is an
  independent set in the complement graph $\widebar{G_{f,H}}$, so the
  clique cover number of $G_{f,H}$ equals the chromatic number of
  $\widebar{G_{f,H}}$.  Also note that $\widebar{G_{f,H}}$ is itself a
  $d$-dependent random graph, with edge probability $q = 1-p^2$.
  Therefore, by Theorem~\ref{thm:chromatic}, $\Pr[BAD_f] <
  \exp(-n^{1+\eps})$.  Finally, let $BAD \deq \left(\bigvee_i
  BAD_i\right)\bigvee\left(\bigvee_{\mbox{d-limited } f}
  BAD_f\right)$.  There are $n$ indices $i$ and at most $n^n \leq
  \exp(n\log n)$ $d$-limited functions $f$.  Therefore, buy a union
  bound we have $$\Pr[BAD] < n\Pr[BAD_i] + n^n \Pr[BAD_f] < n
  e^{-\frac{np^2}{2}} + n^n e^{-n^{1+\eps}} < 1 .$$ Therefore, there
  exists a good $H$.  Also note that in $\cP_H$ for a good $H$, Alice
  and Bob each communicate $O(n\frac{\log \log n}{\log n})$ bits.  This
  completes the proof.
\end{proofof}

\paragraph{Simultaneous Messages.}
%
We conclude this section by showing how to convert
$\cP_H$ into an SM protocol.  Observe that Carol selects a bit from
Alice's message (namely, the clique containing $f(i)$) and a few bits
from Bob's message (the neighbors of $i$ in $H$) and \xor s them
together.  To convert $\cP_H$ to an SM protocol, Alice and Bob send
the same messages as in $\cP_H$.  Carol, given $i$ and $f$, sends a
bitmask describing which bit from Alice's message and which bits from
Bob's message are relevant.  The Referee then \xor s these bits
together, again producing $\mpj_3(i,f,x)$.  Carol sends one bit for
each bit of communication sent by Alice and Bob.  Thus, this SM
protocol costs twice as much as the cost of $\cP_H$.  
We get the following result.

\begin{lemma}[Restatement of Lemma~\ref{lem:mpj-sm-ub}]
  $\smcc(\mpj_3) = O(n\frac{\log \log n}{\log n})$.
\end{lemma}

\section{Proofs of Main Technical Lemmas}\label{sec:maintechnical}
In this section, we state and prove three technical lemmas which form
key insights to our contribution.  The first lemma states that most
sets of $k$ vertices ``look independent''.  The second bounds the
expected number of \emph{intersecting} $k$-cliques.  The final lemma
gives a lower bound on the expected number of \emph{disjoint}
\independent $k$-cliques.

We remind the reader that all three lemmas apply to arbitrary
$d$-dependent random graph distributions.

\begin{lemma}[Restatement of Lemma~\ref{lem:independent}]\label{lem:independent2}
  Suppose $d$ and $k$ are integers such that $dk^3 \leq n$.
  Fix any $d$-dependent graph $\ldgraph{n}{p}{d}$, and let
  $S$ be a set of $k$ vertices uniformly chosen from $V$.  Then, we have
  $$\Pr[S \mbox{ is \independent}] \geq 1-\frac{3dk^3}{2n}\ .$$
\end{lemma}
\begin{proof}
  We divide the possible conflicts into two classes, bound the
  probability of each, and use a union bound.  Say that correlated
  edges are \emph{local} if they share a vertex.  Otherwise, call them
  \emph{remote}.  Let $\cL$ and $\cR$ be the events that $S$ contains
  a local and remote dependency respectively.

  First, we bound $\Pr[\cR]$.  Imagine building $S$ by picking
  vertices $v_1,\ldots, v_k$ one at a time uniformly.  Let $S_i \deq
  \{v_1,\ldots, v_i\}$, and let $B_i$ be the the set of vertices that
  would create a remote dependency if added to $S_i$.  Note that $B_1
  = \emptyset$ since there are no edges in $S_1$ (it contains only one
  vertex).  More importantly, for $i>1$, there are at most ${i\choose
    2}\cdot (2d) < di^2$ vertices in $B_i$, because $S_i$ contains ${i\choose
    2}$ edges; each edge depends on at most $d$ other edges, and each
  of these edges contributes at most two vertices to $B_i$.  It
  follows that $\cR$ is avoided if $v_{i+1} \not\in B_i$ for each
  $i=2\ldots k-1$.  There are $(n-i)$ choices for $v_{i+1}$, so
  $$\Pr[\lnot \cR] \geq
  \prod_{i=2}^{k-1}\left(1-\frac{di^2}{n-i}\right) \geq
  \left(1-\frac{dk^2}{n-k}\right)^{k-2} \geq 1-\frac{dk^3}{n}\ ,$$
  Hence $\Pr[\cR] \leq dk^3/n$.  At first glance, it might appear like
  we've handled local dependencies as well.  However, it is possible
  that when adding $v_i$, we add local dependent edges, if these
  edges are both adjacent to $v_i$.  Thus, we handle this case
  separately.

  Let $\cL_{ij}$ denote the event that $i,j\in S$ and there are no
  local dependencies in $S$ involving $(i,j)$.  Call a vertex $\ell$
  bad for $(i,j)$ if either $(i,\ell)$ or $(j,\ell)$ depend on
  $(i,j)$.  There are at most $d$ bad vertices for $(i,j)$.  Note that
  $\Pr[i,j\in S] = {n-2\choose k-2}/{n \choose k} = k(k-1)/n(n-1)$ and that 
  \begin{align*}    
    \Pr[\lnot \cL_{ij} | i,j \in S] &\geq {n-2-d \choose k-2}/{n-2
      \choose k-2} \\
    & \geq \prod_{z=0}^{d-1} \left(1-\frac{k-2}{n-2-z}\right) \\
    & \geq \left(1-\frac{k-2}{n-2-d}\right)^d \\
    & \geq 1-\frac{d(k-2)}{n-2-d} \\
    & \geq 1-\frac{dk}{n}\ .
  \end{align*}
  It follows that $\Pr[\cL_{ij}] = \Pr[i,j \in S]\Pr[\cL_{ij} | i,j
    \in S] \leq \frac{k(k-1)}{n(n-1)} \cdot \frac{dk}{n}$.  There are
  ${n \choose 2}$ possible pairs $i,j$, so by a union bound, we have
  $\Pr[\cL] \leq \frac{n(n-1)}{2}\frac{k(k-1)}{n(n-1)}\frac{dk}{n}
  \leq \frac{dk^3}{2n}$.  Another union bound on $\cR$ and $\cL$
  completes the lemma.
\end{proof}

\begin{lemma}\label{claim:W-ub2}
  Let $d,p,k$ be such that $k < \frac{\log(n/(2d\log^3n))}{\log 1/p}$.
  Fix a $d$-dependent random graph distribution $\gnpd$.  Let
  $G\sim \gnpd$, and let $W$ be the set of ordered pairs $(S,T)$ such
  that $S,T$ are intersecting \independent $k$-cliques.  Then, $$E[|W|]
  \leq 2k{n\choose k}p^{2{k\choose 2}-1}{k \choose 2}{n \choose
    2}\ .$$
\end{lemma}
\textbf{Note:} To understand the relationship between $d,k,p,n$, it is
helpful to consider the case $d=n^{o(1)}$.  In this setting, the lemma
holds as long as $k \leq (1-o(1))\frac{\log n}{\log 1/p}$.
\begin{proof}
  Let $S,T$ be arbitrary sets of $k$ vertices, and let $X = S\cap T$.
  We calculate $E[|W|]$ by iterating over all possible values of $S,X$
  and for each pair, counting the expected number of $T$ such that
  $S\cap T = X$ and $S,T$ are both $k$-cliques.  For $S,X$, let
  $F(S,X)$ be the expected number of \independent $k$-cliques $T$ such
  that $S\cap T = X$, conditioned on $S$ being a $k$-clique.  Also let
  $F(\ell)$ be the maximum of all $F(S,X)$, taken over all $S$ and all
  $X\subset S$ with $|X| = \ell$.  We have
  \begin{align}
    E[|W|] &= \sum_S \Pr[S \iskclique]\sum_{X\subset S} \sum_{T: S\cap
      T=X} \Pr[T \iskclique|S \mbox{ is k-clique}] \\ &= \sum_S p^{k
      \choose 2} \sum_{X\subset S} F(S,X) \\ &\leq \sum_{S} p^{{k
        \choose 2}} \sum_{\ell=2}^{k-1} \sum_{\substack{X\subset
        S\\ |X| = \ell}} F(\ell) \\ &\leq {n \choose k}p^{{k \choose
        2}}\sum_\ell {k \choose \ell}F(\ell)\ .\label{eqn:EW}
  \end{align}
  Next, we obtain an upper bound on $F(\ell)$.  Since we need only an
  upper bound, we take a very pessimistic approach.  Let $M \subset
  [n]\sm S$ be the set of vertices adjacent to an edge $e$ that
  depends on some edge from $\SnotX{S}{X}$.  Each edge in
  $\SnotX{S}{X}$ depends on at most $d$ other edges, and there are ${k
    \choose 2} - {\ell \choose 2}$ edges in $\SnotX{S}{X}$.
  Therefore, $|M| \leq d({k\choose 2} - {\ell \choose 2})$.  Now, let
  $E(M)$ be the set of edges with one endpoint in $M$ and the other
  endpoint in $M\cup X$.  Each of these edges may be correlated with
  edges in $\SnotX{S}{X}$, so for any $e \in E(M)$ we assume only
  $\Pr[e|S \iskclique] \leq 1$.  On the other hand, by construction
  any edge $e$ \emph{not} in $E(M)$ is independent of $S$, and
  therefore $\Pr[e \in G|S \iskclique] = p$.  Next, we sum over all
  possible $T$, grouping by how much $T$ instersects $M$.  Suppose
  $|T\cap M| = \ell'$ for some $0 \leq \ell' \leq k-\ell$.  Then, $T$
  contains ${k \choose 2}$ edges, ${\ell \choose 2}$ of these edges
  have both endpoints in $X$, and are fixed after conditioning on $S$
  being a $k$-clique.  Of the remaining edges, $\ell\cdot\ell' +
  {\ell' \choose 2}$ are in $E(M)$; the rest are independent of $S$.
  Thus, when $|T\cap M| = \ell'$, then $\Pr[T \iskclique | S
    \iskclique] \leq p^{{k \choose 2} - {\ell \choose 2} - \ell \ell'
    - {\ell' \choose 2}}$.

  \begin{align}
    F(\ell) &= \sum_{T:S\cap T = X} \Pr[T \iskclique | S \iskclique]
    \\ &= \sum_{\ell'=0}^{k-\ell} \sum_{\substack{T: S\cap T = X\\|T
        \cap M| = \ell'}} \Pr[T \iskclique | S \iskclique] \\ &\leq
    \sum_{\ell'=0}^{k-\ell} {M\choose \ell'}{n-k-M \choose k-\ell
      -\ell'} p^{{k \choose 2} - {\ell \choose 2} - \ell \ell' -
      {\ell' \choose 2}} \\ &= p^{{k \choose 2} - {\ell \choose
        2}}\sum_{\ell'=0}^{k-\ell} F^*(\ell')\ ,\label{eq:fstar}
  \end{align}
  where $F^*(\ell') \deq {M\choose \ell'}{n-k-M \choose k-\ell -\ell'}
  p^{- \ell \ell' - {\ell' \choose 2}}$.  Next, we show that the
  summation in Equation (\ref{eq:fstar}) telescopes.

  \begin{claim}\label{claim:fstar-telescope}
    If $k \leq \frac{\log\left(\frac{n}{2d\log^3 n}\right)}{\log
      1/p}$ then $\sum_{\ell' = 0}^{k-1} F^*(\ell') \leq 2F^*(0)$.
  \end{claim}
  \begin{proof}
    Fix any $0 \leq i < k-\ell$, and consider $F^*(i+1)/F^*(i)$.
    Using ${a \choose b+1}/{a \choose b} = \frac{a-b}{b+1}$ and ${a
      \choose b-1}/{a \choose b} = \frac{b}{a-b-1}$ and recalling that
    $M < d{k\choose 2}$, we have:
    \begin{align*}
      \frac{F^*(i+1)}{F^*(i)} &= \frac{{M\choose i+1}{n-k-M \choose
          k-(i+1)}p^{- \ell (i+1) -(i+1)i/2}}{{M \choose i}{n-k-M
          \choose k-i}p^{-\ell i - i(i-1)/2}} \\ &=
      \frac{M-1}{i+1}\frac{k-i}{n-k-M-k+i}p^{-\ell -i} \\ &\leq
      \frac{dk^2}{2}\frac{k}{n-o(n)}\left(\frac{1}{p}\right)^k \\ &<
      \frac{dk^3}{n} \left(\frac{1}{p}\right)^k \\ &<
      \frac{k^3}{2\log^3 n} \\ &< 1/2\ ,
    \end{align*}
    where the penultimate inequality holds because of our assumption
    on $k$, and the final inequality holds because $k < \log n$.
    We've shown that for all $i$, $F^*(i+1)/F^*(i) < 1/2$.  Hence
    $F^*(i) < F^*(0)2^{-i}$, and so $\sum_{\ell'} F^*(\ell') \leq
    \sum_{\ell'}F^*(0)2^{-\ell'} \leq 2F^*(0)$.
  \end{proof}
  From claim~\ref{claim:fstar-telescope}, we see that $$F(\ell) \leq
  p^{{k\choose 2}-{\ell \choose 2}} \sum_{\ell' = 0}^{k-\ell}
  F^*(\ell') \leq 2p^{{k\choose 2}-{\ell \choose 2}} F^*(0) =
  2p^{{k\choose 2}-{\ell \choose 2}} {n-k-M \choose k-\ell}\ .$$ Now,
  plugging this inequality back into Equation~\ref{eqn:EW}, we get
  $$E[|W|] \leq {n\choose k}p^{k \choose 2}\sum_\ell {k \choose
    \ell}F(\ell) \leq 2{n \choose k}p^{k \choose 2}\sum_\ell {k
    \choose \ell} p^{{k \choose 2}-{\ell \choose 2}}{n-k-M \choose
    k-\ell}\ .\label{eqn:partialW}$$
    Let $G(\ell) \deq p^{{k \choose 2} - {\ell \choose 2}}{k \choose
      \ell}{n-k-M \choose k-\ell}$, and for $2 \leq \ell < k-1$, let
    $G^*(\ell) \deq G(\ell)/G(\ell+1)$.  Note that $$G^*(\ell) =
    p^\ell \frac{\ell+1}{k-\ell}\frac{n-2k-M+\ell+1}{k-\ell}\ .$$ We
    claim that $G^*(\ell)$ decreases as long as $p < 8/27-\Omega(1)$.
    To see this, note that $$\frac{G^*(\ell)}{G^*(\ell+1)} = p
    \frac{\ell+1}{\ell} \cdot \left( \frac{k-\ell+1}{k-\ell} \right)^2
    \frac{n-2k-M+\ell+1}{n-2k-M+\ell} < p (3/2)^3(1+o(1))\ ,$$ where
    the inequality holds because $(a+1)/a = 1+1/a$ and because $\ell,
    k-\ell \geq 2$ for the range of $\ell$ we need when calculating
    $G^*(\ell)$.  In a way, saying that $G^*(\ell)$ is decreasing
    amounts to saying that $G(\ell)$ is convex---once $G(i) \leq
    G(i+1)$, then $G(j) \leq G(j+1)$ for all $j>i$.  Next, a
    straightforward calculation using our choice of $k$ shows that
    $G(k-1) \leq G(2)$.  Thus, it must be the case that $G(i) \leq
    G(2)$ for all $i$, and therefore $$E[|W|] \leq 2{n\choose k}p^{k
      \choose 2}kG(2) = 2k{n \choose k}p^{2{k \choose 2} - 1}{k
      \choose 2}{n-k-M \choose k-2} < 2k{n \choose k}p^{2{k\choose
        2}-1}{k \choose 2}{n \choose k-2}\ .$$ This completes the
    proof of Lemma~\ref{claim:W-ub2}.
\end{proof}
    
Finally, we prove the lemma that in any $d$-dependent graph
distribution, the expected number of \emph{disjoint} \independent
$k$-cliques is large.  Recall that $Y$ is the maximal number of
disjoint \independent $k$-cliques.
\begin{lemma}[Restatement of Lemma~\ref{lem:Y}]
  $\mathbb{E}[Y] \geq \frac{n^2p}{19k^5}.$
\end{lemma}
\begin{proof}
    We construct $Y$ probabilistically, by selecting each potential
    \independent $k$-clique with small probability and removing any
    pairs of $k$-cliques that intersect.
%
%
    Let $K$ denote the family of \independent $k$-cliques.  By
    Lemma~\ref{lem:independent} and our choice of $d$, a randomly
    chosen set $S$ of $k$ vertices is \independent with probability at
    least $2/3$.
    By this and our choice of $k$, we have 
    $$\E[\vert K \vert] \geq \frac{2}{3}\binom{n}{k}
    p^{\binom{k}{2}}\ .$$
    Recall that $W$ is the set of ordered pairs $\{S,T\}$ of \independent
    $k$-cliques such that $ 2 \leq \vert S \cap T \vert < k.$ For our
    argument, we require an upper bound on $\E[|W|]$.  In the standard
    random graph model, if $|S\cap T| = \ell$, then $\Pr[S,T \mbox{
        both k-cliques}] = p^{{k \choose 2} - {\ell \choose 2}}$.
    However, this no longer holds for $d$-dependent distrubtions, even
    if $S$ and $T$ are both \independent.  This is because while edges
    in $S$ and $T$ are independent, edges in $S$ but not $T$ may be
    correlated with edges in $T$ but not $S$.  As an extreme case,
    suppose all edges in $S$ are independent, but each edge in $S\sm
    T$ is completely correlated with an edge in $T\sm S$.  Then,
    $\Pr[S,T \mbox{ k-cliques}] = \Pr[S \iskclique] = \Pr[T
      \iskclique] = p^{{k\choose 2}}$. Essentially, allowing edges to
    be correlated has the potential to drive up the variance on the
    number of $k$-cliques, even when these $k$-cliques are
    \independent.  This is perhaps to be expected.  Nevertheless, in
    Lemma~\ref{claim:W-ub2}, we were able to show that when $d$ is
    small, this increase is not much more than in the standard graph
    model.

With this claim, we are now able to construct a large set of disjoint
\independent $k$-cliques with high probability.  Create
$K'\subseteq K$ by selecting each uncorrelated $S\in K$ independently
with probability $$\Pr[S \in K'] = \gamma = \frac{1}{12kp^{{k\choose
      2}-1}{k\choose 2}{n \choose k-2}}\ .$$ Finally, create $L$ from
$K'$ by removing each pair $S,T \in K'$ such that ${S,T} \in W$.  By
construction, $L$ is a set of edge-disjoint \independent $k$-cliques;
furthermore, we have
    \begin{align*}
      E[|L|] &= \gamma E[|K|] - 2\gamma^2 E[|W|] \\ &\geq
      \frac{2\gamma}{3}{n\choose k}p^{{k \choose 2}} - \frac{2\gamma
        \cdot2k {n\choose k}p^{2{k\choose 2}-1}{k\choose 2}{n\choose
          k-2}}{12kp^{{k\choose 2}-1}{k\choose 2}{n \choose k-2}}
      \\ &= \frac{2\gamma}{3}{n\choose k}p^{{k\choose
          2}}-\frac{\gamma}{3}{n\choose k}p^{{k\choose 2}} \\ &=
      \frac{\gamma}{3}{n\choose k}p^{k\choose 2} \\ &= \frac{{n\choose
          k}p^{k\choose 2}}{3\cdot 12k p^{{k\choose 2}-1} {k\choose
          2}{n\choose 2}} \\ &\geq \frac{{n\choose k}}{{n\choose k-2}}
      \frac{p}{36k} \frac{1}{{k\choose 2}} \\ &\geq \frac{p}{18k^3}
      \frac{{n\choose k}}{{n\choose k-2}} \\ &= \frac{p}{18k^3}
      \frac{(n-k-2)(n-k-1)}{k(k-1)} \\ &\geq
      \frac{p}{18k^3}\frac{18n^2}{19k^2} \\ &= \frac{n^2 p}{19k^5} \ ,
    \end{align*}
    where the final inequality holds for large enough $n$.
\end{proof}

\section{Results for Non-Boolean Pointer Jumping}\label{sec:nonbool}
In this section, we leverage the protocol for $\mpj_3$ to achieve new
results for the non-Boolean Pointer Jumping problem $\mpjh$.  Let $\cQ$ be the
protocol for $\mpj_3$ given in Lemma~\ref{lem:mpj-sm-ub}.
First, we give a protocol for $\mpjh_3$.  The cost matches the upper
bound from~\cite{DammJS98} but has the advantange of working in the
Simultaneous Messages model.
\begin{lemma}[Restatement of Lemma~\ref{lem:mpjh-sm-ub}]
  There is an $O(n \log \log n)$-bit SM protocol for $\mpjh_3$.
\end{lemma}
\begin{proof}
  Run $\cQ$ $\log n$ times in parallel, on inputs $(i,f_2,z_1),
  (i,f_2, z_2), \ldots, (i, f_2, z_{\log n})$, where $z_j$ denotes the
  $j$th most significant bit of $f_3$.  This allows the Referee to
  recover each bit of $f_3(f_2(i)) = \mpjh(i,f_2,f_3)$.
\end{proof}

Next we give a new upper bound for $\mpjh_4$.  As far as we know, this
is the first protocol for $\mpjh_k$ for any $k$ that uses a sublinear
amount of communication.
\begin{theorem}[Restatement of Theorem~\ref{thm:mpjh-ub}]
  There is a one-way protocol for $\mpjh_4$ with cost $O(n \frac{(\log
    \log n)^2}{\log n})$.
\end{theorem}
\begin{proof}
  Let $i,f_2,f_3,f_4$ be the inputs to $\mpjh_4$, and for $1 \leq j
  \leq \log n$, let $z_j \in \{0,1\}^n$ be the string obtained by
  taking the $j$th most significant bit of each $f_3(w)$ (i.e.,
  $z_j[w]$ is the $j$th most significant bit of $f_3(w)$.)  Fix a
  parameter $k$ to be determined shortly.  $\plr_1,\plr_2$, and
  $\plr_3$ run $\cQ$ on $\{(i,f_2,z_j) : 1 \leq j \leq k\}$.  From
  this, $\plr_3$ learns the first $k$ bits of $f_3(f_2(i))$.  She then
  sends $f_4(z)$ for every $z \in \{0,1\}^{\log n}$ whose $k$ most
  significant bits match those of $f_3(f_2(i))$.  $\plr_4$ sees $i,
  f_2,$ and $f_3$, computes $z^* \deq f_3(f_2(i))$, and recovers
  $f_4(z^*)$ from $\plr_3$'s message.  Note that there are $n/2^k$
  strings that agree on the first $k$ bits, and for each of these
  strings, $\plr_3$ sends $\log n$ bits.  Therefore, the cost of this
  protocol is $k\cost(\cQ) + n\log(n)/2^k = O\left(k n \frac{\log \log
    n}{\log n} + n\log(n) 2^{-k}\right)$.  Setting $k \deq 2\log
  \frac{\ln 2 \log n}{\log \log n} = \Theta(\log \log n)$ minimizes
  the communication cost, giving a protocol with cost $O\left(n
  \frac{(\log \log n)^2}{\log n}\right)$.
\end{proof}

\section{Dependent Graphs with Large Cliques or Large Dependency} 
\label{sec:otherstuff}
In this section, we provide results that witness the tightness of our
current bounds.  The next lemma shows that there exist dependent
random graphs that almost surely contain cliques of size $\Omega(d)$, and others that almost surely have cliques of size $\Omega(\sqrt{d}\log(n))$.
\begin{lemma}{(Restatement of Lemma~\ref{lem:large-clique})}
  For all constant $0 < p < 1$ and $d = o(n)$,
  \begin{enumerate}
  \item
    there exists a $d$-dependent random
    graph $\gnpd$ such that $$\Pr\left[\clique(\gnpd) >
      \frac{d\sqrt{p}}{2} - d^{\frac{1}{2}}p^{\frac{1}{4}}\right] >
    1- e^{-2n/d}\ .$$
  \item
    there exists a $d$-dependent random graph $\gnpd$ such that almost
    surely $$\clique(\gnpd) =
    \Omega(\sqrt{d}\log(n))\ .$$
  \end{enumerate}
\end{lemma}
\begin{proof}
  We give two constructions.

  For the first result, fix $d' \deq \frac{d\sqrt{p}}{2} -
  \sqrt{d\sqrt{p}}$ and $M_1 \deq 2n/d$.  Partition the vertices into
  $M_1$ sets $V_1,\ldots, V_{M_1}$ each of size $d/2$.  Let $c(i)$ denote
  the part containing $i$ (we think of $i$ has having color $c$).
  Now, let $\{X_{i,c} : i \in V, 1\leq c \leq M_1\}$ be a series of
  i.i.d.~random bits with $\Pr[X_{i,c} = 1] = \sqrt{p}$, and place
  $(i,j) \in \gnpd$ if $X_{i, c(j)} \bigwedge X_{j, c(i)} = 1$.  Thus,
  $(i,j)$ is an edge with probability $p$.  Also note that edges
  $(i,j)$ and $(i',j')$ are dependent if either $c(i) = c(i')$ or
  $c(j) = c(j')$.  Since there are $d/2$ vertices in each $V_\ell$,
  $(i,j)$ is dependent on at most $d$ other edges and $\gnpd$ is
  $d$-dependent.

  Now, fix a color $c$, and let $S_c \deq \{i : c(i) = c \wedge
  X_{i,c} = 1\}$.  For any $i,j \in S_c$ we have $X_{i,c} = X_{j,c} =
  1$ and that $c(i) = c(j) = c$.  Therefore, $(i,j) \in \gnpd$ for any
  $i,j \in S_c$, hence $S_c$ is a clique.

  Next, consider $|S_c|$.  There are $d/2$ vertices with color $c$, so
  $E[|S_c|] = \frac{d\sqrt{p}}{2}$.  By the Chernoff bound, $\Pr[|S_c|
    < d'] < \frac{1}{e}$, so the probability that there is some color
  $c$ with $|S_c| \geq d'$ is at least $1-e^{-2n/d}$.  Therefore,
  $\gnpd$ almost surely contains a clique of size at least $d'$.

  For the second graph, partition the vertices $[n]$ into $M_2 \deq
  n/\sqrt{d}$ subsets $V_1, \ldots, V_{M_2}$, each of size $\sqrt{d}$.
  Let $c(i)$ be the subset containing $i$.  Let $\{X_{c_1,c_2}: 1 \leq
  c_1,c_2 \leq M_2\}$ be a set of independent, identically distributed
  binary variables with $\Pr[X_{c,c'} = 1] = p$.  Now, place edge
  $(i,j)$ in the graph if $X_{c(i),c(j)} = 1$.  In this way, for any
  $V_s, V_t$, either all edges between $V_s$ and $V_t$ exist, or none
  do, and similarly for any $V_s$, either all edges between vertices
  in $V_s$ will be in the graph, or none will.

  Next, let $S$ be the set of all $i$ such that edges between
  vertices in $V_i$ are in the graph.  Each $i \in S$ with probability
  $p$.  By standard Chernoff bounds, $|S| \geq pM_2/2$ with
  high probability.  Let $M' \deq pM_2/2$.
  The construction above induces a new random graph $G'$ on $M'$
  vertices where all edges are i.i.d. in $G'$ with probablity $p$.
  i.e., $G'$ is an \erdosrenyi random graph on $M'$ vertices.
  By~\cite{Bollobas76}, $\clique(G') \geq 2\log(M')/\log(1/p) =
  \Omega(\log(n)/\log(1/p))$ with high probability.  Finally, a clique
  of size $k$ in $G'$ gives a clique of size $k\sqrt{d}$ in $G$, hence
  $G$ contains a clique of size $\Omega(\sqrt{d}\log(n)/\log(1/p))$
  with high probablity.
\end{proof}
Our second result in this section shows that when the dependency
factor becomes $\Omega(n)$, essentially nothing can be said about the clique number of dependent random graphs.
\begin{lemma}{(Restatement of Lemma~\ref{lem:larged})}
  Fix $d\deq 2n-2$. Then, the following statements hold.
  \begin{enumerate}
  \item For any $0<p<1$, there exists a $d$-dependent random
    graph $\gnpd$ that is bipartite with certainty.
  \item For any $1/2 \leq p < 1$, there exists a $d$-dependent random
    graph $\gnpd$ such that $\clique(\gnpd) \geq n/2$ with certainty.
  \end{enumerate}
\end{lemma}
\begin{proof}
  We again provide two constructions.  For the first construction, set
  $q_1\deq 1-\sqrt{1-p}$, and let $X_1,\ldots, X_n$ be i.i.d. random
  bits such that $X_i =1$ with probability $q_1$.  Think of each
  $X_i$ as being assinged to vertex $v_i$.  Now, place edge $(i,j) \in
  \gnpd$ iff $X_i \oplus X_j = 1$.  Note that $(i,j)\in \gnpd$ with
  probability $2q(1-q) = p$.  It is easy to see that $(i,j)$ depends
  on $(i',j')$ only if either $i=i'$ or $j=j'$.  There are at most
  $2(n-1)$ such edges, hence the random graph is $d$-dependent.
  Finally, we claim that the graph is bipartite.  To see this, suppose
  for the sake of contradiction that $\gnpd$ contains an odd cycle
  $(1,2, \ldots, 2k+1, 1)$.  Without loss of generality, assume that
  $X_1=1$ (the proof is similar if $X_1=0$.)  Since each edge $(i,i+1)
  \in \gnpd$, we must have that $X_2,X_4,\ldots, X_{2k}$ all equal
  $0$, and $X_1,X_3,\ldots, X_{2k+1}$ all equal $1$.  But then
  $X_1=X_{2k+1} = 1$, hence $(1,2k+1) \not\in \gnpd$.  This
  contradicts the assumption that $(1,2,\ldots, 2k+1,1)$ is a cycle.

  We proceed with the second construction in a similar manner.  Let
  $q_2 \deq \frac{1}{2}(1-\sqrt{2p-1})$, and let $X_1,\ldots, X_n$ be
  i.i.d. random bits with $\Pr[X_i = 1] = q_2$.  This time, place
  $(i,j) \in \gnpd$ iff $X_i = X_j$.  Note that $(i,j)$ is an edge
  with probablity $q_2^2 + (1-q_2)^2 = p$.  Now, let $S_0 \deq
  \{i:X_i=0\}$ and similarly $S_1 \deq \{i:X_i=1\}$.  It is easy to
  see that $S_0$ and $S_1$ are both cliques in $\gnpd$.  One of them must
  contain at least half the vertices.
\end{proof}

\bibliographystyle{plain}
\bibliography{super}

\end{document}